\documentclass[a4paper]{article}

\usepackage{amsfonts}
\usepackage{amssymb}
\usepackage{graphicx}
\usepackage{amsmath}
\usepackage[a4paper]{geometry}

\newtheorem{theorem}{Theorem}

\newtheorem{condition}{Condition}

\newtheorem{definition}{Definition}

\newtheorem{lemma}{Lemma}

\newtheorem{proposition}{Proposition}
\newtheorem{remark}{Remark}

\newenvironment{proof}[1][Proof]{\noindent\textbf{#1.} }{\ \rule{0.5em}{0.5em}}

\begin{document}

\title{Locally Adaptive Density Estimation on the Unit Sphere Using Needlets}

\bigskip

\author{\textsc{Audrey Kueh}\footnote{Statistical Laboratory, Department of Pure Mathematics and Mathematical Statistics, University of Cambridge. Wilberforce Road CB3 0WB, Cambridge, UK. This forms part of the author's PHD thesis written under the supervision of Richard Nickl, whose expertise, understanding, and patience helped shape this paper. I also have much gratitude to the referees who wrote such extensive reports helping me to improve this paper. All mistakes remain mine.}}

\date{November 2011}

\maketitle

\begin{abstract}
The problem of estimating a probability density function $f$ on the $d-1$-dimensional unit sphere $S^{d-1}$ from directional data using the needlet frame is considered. It is shown that the decay of needlet coefficients supported near a point $x \in S^{d-1}$ of a function $f: S^{d-1} \to \mathbb R$ depends only on local H\"{o}lder continuity properties of $f$ at $x$. This is then used to show that the thresholded needlet estimator introduced in Baldi, Kerkyacharian, Marinucci and Picard \cite{Baldi} adapts to the local regularity properties of $f$. Moreover an adaptive confidence interval for $f$ based on the thresholded needlet estimator is proposed, which is asymptotically honest over suitable classes of locally H\"{o}lderian densities.
\end{abstract}

\textbf{MSC Classification Numbers:} 42C40, 60E15, 62G07

\section{Introduction}
Let $S^{d-1}$ denote the surface of the unit sphere in $\mathbb{R}^d$ and let $X = \{X_1,\,X_2,\ldots,\,X_n\}$ be an independent and identically distributed sample of $n$ values from some probability density function $f: S^{d-1} \rightarrow \mathbb{R}$. Our goal is to estimate $f$ from the sample. A classical method for doing this is by kernel density estimation, see \cite{cross, plug}. In the past few years functions known as needlets have been constructed, which have given us a powerful new tool to tackle this problem. Needlets are effectively built on the spherical harmonics to form a tight frame for the space $L^2(S^{d-1})$ of square-integrable functions on $S^{d-1}$, in such a way that they have a localised projection kernel (see \cite{decomp, local}). Baldi, Kerkyacharian, Marinucci and Picard \cite{Baldi} have shown how to use the needlet frame combined with the standard thresholding techniques to construct an estimator that achieves the minimax convergence rates (up to $\log$ terms) over the usual Besov spaces in the $L^p$ norms, $1 \le p \le \infty$. These Besov spaces contain functions which can be approximated well by spherical polynomials globally on $S^{d-1}$ in the $L^p$ norms, and thus model homogeneous smoothness properties of functions on $S^{d-1}$.

However these results do not address spatially inhomogeneous smoothness properties of $f$, for which there are currently no results in the literature for functions on $S^{d-1}$ (although results for functions on the real line exist as in \cite{Jaffard2, Lepski}). For example, the density $f$ could be $t$-differentiable except for a single point $y$ where it behaves locally like $d(x, y)^{t'}$ for some noninteger $t' < t$, $d$ being the geodesic distance on $S^{d-1}$, which means that the function cannot be $t$-differentiable globally. We thus define $f: S^{d-1} \to \mathbb R$ to be locally $t$-H\"{o}lder-continuous at $x \in S^{d-1}$ if it can be approximated locally by a $\lfloor t \rfloor$-th degree polynomial with suitable error bounds, as suggested by Jaffard \cite{Jaffard2}. We show that the local needlet coefficients and the approximation errors from the corresponding local needlet projections of such functions obey decay properties that reflect only the pointwise H\"{o}lderian regularity properties -- a result that does not follow from the global needlet characterisation of Besov spaces on $S^{d-1}$ obtained in \cite{decomp}. This is the analogue of the results for wavelets by Andersson \cite{char}. This can be used to prove explicitly that the thresholded needlet estimator from \cite{Baldi} is locally minimax-optimal within logarithmic factors. This is the subject of our first result, Theorem~\ref{crazycheese}. We note here that these logarithmic factors are probably necessary for they also appear (and are necessary) in estimations of densities on the real line \cite{Lepski} \cite{first2}.

A next challenge is to construct a confidence interval for the unknown function $f$ at $x \in S^{d-1}$. One can use Bernstein's inequality to find a confidence interval for each unknown needlet coefficient, centered at the corresponding empirical needlet coefficient. To create a confidence interval for $f$ centered at the thresholded needlet estimator we create a confidence interval around each \textit{non-thresholded coefficient} and sum the result. This procedure is shown to give confidence intervals of adaptive expected length. Proving coverage will thus require some assumptions -- it is well known from Low \cite{Low} that adaptive and honest confidence intervals cannot exist over the usual smoothness classes. The method indicated by \cite{nickl, winner, more} is to assume that the underlying function satisfies a further lower bound condition on the local decay of the needlet coefficients. Indeed we show in Theorem~\ref{FINISHHIM} that the proposed confidence interval is asymptotically honest over functions satisfying this condition. However, in contrast with these results, our confidence intervals are spatially adaptive. The practical implementation of this estimator is beyond the scope of this paper; it is hoped that a further paper will be published on this.

The format of the paper is as follows: We first summarise the construction of needlets in Section~\ref{construction}. Then, we define the local regularity spaces $C^t_{M,\delta}(x)$ and show that the needlet coefficients of such functions obey appropriate decay properties. We also find a lower bound for the minimax error of pointwise estimators for densities in these spaces. Following that, we introduce the hard thresholded estimator and use these properties to show it is locally near-minimax-optimal and create asymptotic confidence intervals around it. Finally all proofs will be given at the end.

\section{Needlets and Local Regularity Properties of Functions}
\subsection{The Needlet Frame}
\label{construction}
We first start with a review of spherical harmonics; more details can be found in Stein and Weiss \cite{Stein} or in Faraut \cite{faraut}. Let $L^2(S^{d-1})$ be the space of square integrable functions on $S^{d-1}$ with the natural inner product $$<f,g> = \int\limits_{S^{d-1}} f(x) g(x)\,dx$$ where $dx$ is normalised such that $\int_{S^{d-1}}dx$ equals the Lebesgue measure $\omega_{d-1}$ of $S^{d-1}$. Let further $H_k(S^{d-1})$ be the space of spherical harmonics of degree $k$. Then:
\begin{equation*}
L^2(S^{d-1}) = \bigoplus_{k \geq 0} H_k(S^{d-1})
\end{equation*}
with convergence in $L^2(S^{d-1})$. If we denote the projector kernels onto the spaces $\{H_k\}_{k\geq 0}$ to be $\{Z_k(\cdot,\cdot)\}_{k\geq 0}$, we thus have:
\begin{equation*}
f(x) = \sum_{k=0}^{\infty} \int_{S^{d-1}} Z^k(x,y) f(y)\,dy
\end{equation*}

with convergence in $L^2(S^{d-1})$. It is well known that $$Z^k(x,y) = \frac{2k+d-2}{(d-2)\omega_{d-1}} P_k^{(d-2)/2}(x.y)$$ where $P_k^{(d-2)/2}$ is the corresponding ultraspherical (or Gegenbauer) polynomial and $(\cdot,\cdot)$ is the Euclidean inner product in $\mathbb R^d$. Since for fixed $x$, the kernels $Z^k(x,\cdot)$ are themselves spherical harmonics of degree $k$, we have:
\begin{equation}
\label{ort}
\int_{S^{d-1}} Z^k(x,y) Z^m(z,y)\,dy = \delta_{km} Z^k(x,z)
\end{equation}

The rest of this section follows Narcowich, Petrushev and Ward \cite{decomp, local} who first showed how to construct a needlet frame from the spherical harmonics. We start with a \textit{Littlewood-Paley decomposition}. Let $a$ be a decreasing $C^{\infty}$ function on $\mathbb{R}^+$, compactly supported on $[0,1]$ such that $a(x)=1$ when $x \in \big[0, \frac{1}{2}\big]$. We also define $b = a(\frac{x}{2}) - a(x)$ which is compactly supported on $\big[\frac{1}{2},2\big]$. We define:
\begin{align*}
A_j(f)(x) &:= \int_{S^{d-1}} A_j(x,y)f(y)\,dy \\
A_j(x,y) &:= \sum_k a\Big(\frac{k}{2^j}\Big)Z^k(x,y) \\
B_j(f)(x) &:= \int_{S^{d-1}} B_j(x,y)f(y)\,dy \\
B_j(x,y) &:= \sum_k b\Big(\frac{k}{2^j}\Big)Z^k(x,y) 
\end{align*}

It is obvious that $A_j(f)$ converges to $f$ in the $L^2(S^{d-1})$ norm. More importantly, the kernel $A_j$ can be shown to be localised, see \cite{decomp}, Eqn 1.2: for all $m>0$ there exists $c_m>0$ such that for all $x,y,j,d$:
\begin{equation}
\label{aimpt}
|A_j(x,y)| \leq \frac{c_m 2^{j(d-1)}}{(1+2^jd(x,y))^m}
\end{equation}
where $d(x,y)$ is the geodesic distance between $x,y \in S^{d-1}$. This localisation implies that the error $|A_j(f)(x) - f(x)|$ decays exponentially in $j$ for certain classes of functions, see Proposition~\ref{bias} later on. Now, if we define:
\begin{equation*}
C_j(x,y) := \sum_k \sqrt{b\Big(\frac{k}{2^j}\Big)} Z^k(x,y) = \sum_{2^{j-1}< k<2^{j+1}}\sqrt{b\Big(\frac{k}{2^j}\Big)} Z^k(x,y)
\end{equation*}

Then, by repeated usage of equation~\ref{ort}, we can \textit{split} $B_j$:
\begin{equation*}
B_j(x,y) = \int_{S^{d-1}} C_j(x,u) C_j(y,u) du 
\end{equation*}

Finally, we notice that for fixed $x$ and $y$, $u \rightarrow C_j(x,u) C_j(u,y)$ is a polynomial of degree $2^{j+2}$. By the \textit{quadrature} formulae in Section 4.2 in \cite{local} and Theorem 2.8 in \cite{decomp}, there exists $c>0$, such that for all $j$, there is a set of points $\mathcal{H}_j = \{x_1,x_2\ldots x_k\}$ with $d(x_i,x_j) > c2^{-j}$ for $x_i \neq x_j$ and a set of positive values indexed by the set $\{\lambda_\eta\}_{\eta \in \mathcal{H}_j}$ such that for all polynomials $f$ of degree $2^{j+2}$:

\begin{equation*}
\int_{S^{d-1}} f(x)\,dx = \sum_{\eta \in \mathcal{H}_j} \lambda_{\eta} f(\eta)
\end{equation*}

Hence, we obtain the following expression for $B_j(f)$:
\begin{equation*}
B_j(f)(x) = \int_{S^{d-1}} B_j(x,y)f(y)\,dy = \sum_{\eta \in \mathcal{H}_j} \lambda_{\eta} C_j(x,\eta) \int_{S^{d-1}} C_j(y,\eta)f(y)\,dy
\end{equation*}

This motivates the definition of needlets $\psi_{j\eta} = \sqrt{\lambda_\eta} C_j(., \eta)$ and if $\beta_{i\eta} = <f, \psi_{i\eta}>$, then: 

\begin{equation*}
A_j(f) = \frac{<f, 1>}{\omega_{d-1}} + \sum_{i=0}^{j-1} B_i(f) = \frac{<f, 1>}{\omega_{d-1}} + \sum_{i=0}^{j-1} \sum_{\eta \in \mathcal{H}_i} \beta_{i\eta} \psi_{i\eta}
\end{equation*}

From Corollary 5.3 on \cite{local} we know that the $\psi_{i\eta}(x)$'s are also localised; for all $m>0$ there exists $c_m>0$ such that for all $\eta,y,i,d$:
\begin{equation}
\label{psiimpt}
|\psi_{i\eta}(y)| \leq \frac{c_m 2^{i(d-1)/2}}{(1+2^id(y,\eta))^m}.
\end{equation}

We will need the following lemma, which follows from the previous equation:
\begin{lemma}
\label{salamander}
There exists a constant $C_1$ such that for all $i$, $$\sum_{\eta \in \mathcal{H}_i} |\psi_{i\eta}(y)| \leq C_1 2^{i(d-1)/2}$$
\end{lemma}

\subsection{Regularity spaces of functions}
\label{local}
We will use the standard definition of H\"{o}lder continuity for $t \in [0,1]$:
\begin{definition}
Let $t \in [0,1],\,\delta,\,M \geq 0$ and $x \in S^{d-1}$. We say $f$ is $t$-H\"{o}lder-continuous at $x$ with parameters $M,\,\delta$ if:
\begin{eqnarray*}
\sup_{y \in B(x,\delta), y \ne x} \frac{|f(y) - f(x)|}{d(x,y)^t} \leq M \\
\end{eqnarray*}
We also say that $f \in C_{M,\delta}^t(x)$.
\end{definition}

However, the standard way for defining H\"{o}lder continuity for $t>1$ is to use differentiability and charts. This has proved awkward to work with and so we have gone with an alternative. We first note that functions which are smooth can be approximated by polynomials with quantitative error bounds. If we have a function $g: \mathbb{R} \rightarrow \mathbb{R}$ which is $t$-H\"{o}lder-continuous in a $\delta$ ball around 0, then by the mean value theorem (or Taylor's theorem), for all $x \in B(0,\delta)$ there exists $c \in [0,x]$ such that:
\begin{equation*}
g(x) = g(0) + g'(0)x + \ldots + \frac{g^{\lfloor t\rfloor}(0)}{\lfloor t\rfloor!} x^{\lfloor t\rfloor} + \frac{g^{\lfloor t\rfloor}(c) - g^{\lfloor t\rfloor}(0)}{\lfloor t\rfloor!} x^{\lfloor t\rfloor}
\end{equation*}

This suggests to define a H\"{o}lder norm locally in a $\delta$-neighborhood of 0 as:
\begin{equation*}
||g||_{C^t(B(0,\delta))} = \max_{i=0,1,\ldots,\lfloor t \rfloor} |g^{i}(0)| + \sup_{x \in B(0,\delta)} \left|\frac{g^{\lfloor t\rfloor}(x) - g^{\lfloor t\rfloor}(0)}{x^{t - \lfloor t\rfloor}}\right|
\end{equation*}

We can thus write $g(x) = P(x) + R(x)$, where $P(x)$ is a polynomial of degree $\lfloor t\rfloor$ bounded by $||g||_{C^t(B(0,\delta))}\sum_{i=0}^{\lfloor t\rfloor} \frac{|x|^i}{i!} \leq ||g||_{C^t(B(0,\delta))}e^{|x|}$ and the remainder $R$ is bounded on $B(0,\delta)$ by $x^t||g||_{C^t(B(0,\delta))}$. This definition can easily be generalised to $d$ dimensions, and thus motivates the definition for functions $f: S^{d-1} \rightarrow \mathbb{R}$ to be locally $t$-H\"{o}lder-continuous:

\begin{definition}
\label{lolzcat}
Let $t,\,\delta,\,M \geq 0$ and $x \in S^{d-1}$. We say $f$ is $t$-H\"{o}lder-continuous at $x$ with parameters $M,\,\delta$ if there exists a spherical polynomial $P_f:=P_{f,x}$ of degree $\lfloor t\rfloor$ such that:
\begin{eqnarray*}
\sup_{y \in B(x,\delta), y \ne x} \frac{|f(y) - P_f(y)|}{d(x,y)^t} \leq M \\
||P_f||_\infty \leq M
\end{eqnarray*}
We also say that $f \in C_{M,\delta}^t(x)$.
\end{definition}

For further justification of the above definition, we consider a function $f$ which satisfies the traditional definition: 

\begin{definition}
Let $U \subset {S^{d-1}}$ be a neighbourhood of $x$ and let $C: U \rightarrow R^{d-1}$ be a chart. Then, $f$ is $t$-H\"{o}lder-continuous at $x$ if $f \circ C^{-1}$ is $t$-H\"{o}lder-continuous in a neighbourhood of $C(x)$ for all $C$. 
\end{definition}

We can show that $f$ satisfies Definition~\ref{lolzcat} for some $M$. Without loss of generality, let $x = (1,0,\ldots,0)$ and let $T_x$ be the tangent space attached to $x$; in this example, this is the plane such that $x_1=1$. We then have a local chart from the eastern hemisphere ($x_1>0$) to $T_x$: $C(x_1,x_2,\ldots,x_d) = (1,x_2,\ldots,x_d)$. By our hypothesis, $\hat{f} = f \circ C^{-1}$ is locally $t$-H\"{o}lder-continuous at $B(x,\delta)$ with $\delta$ and $M$ depending on $f$. Now if we extend $\hat{f}$ to $g$ on $\mathbb{B}(x,\delta) \in \mathbb{R}^d$ by $g(x_1,x_2,\ldots,x_d) = \hat{f}(1,x_2,\ldots,x_d)$, $g$ too is locally $t$-H\"{o}lder-continuous at $B(x,\delta)$ with $\delta$ and $M$ depending on $f$, so using Taylor's theorem, $g$ can be estimated in this ball by a polynomial $P$ of degree $\lfloor t \rfloor$ such that $|g(y) - P(y)| \leq Md(x,y)^t$. Now we use the fact that all polynomials restricted to $S^{d-1}$ can be written as a sum of spherical harmonics, see \cite{Stein}, and this yields a $P_f$, which is bounded because it is a continuous function on a compact set. This implies that Definition~\ref{lolzcat} is sensible. 

In contrast to spherical harmonics, the needlet frame allows us to describe local regularity properties of functions $f:S^{d-1} \to \mathbb R$ by the decay of the 'local' needlet series representation of $f$. We give three instances of this fact, all of which shall be useful in what follows. We emphasise that these facts, although not difficult to prove, do not follow from the characterisation of Besov spaces in \cite{decomp}.

\begin{proposition}
\label{bias}
Let $x \in S^{d-1}$ and $f \in C_{M, \delta}^t(x)$. Then there exists $C_2,\,C_3$ only dependent on $M,\, \psi,\, \delta,\, d,\, t,\, ||f||_\infty$ such that for all $i,\,j$:
\begin{eqnarray*}
|A_j(f)(x) - f(x)| \leq C_2 2^{-jt} \\
\sum_{\eta \in \mathcal{H}_i} |\beta_{i\eta} \psi_{i\eta}(x)| \leq C_3 2^{-it}
\end{eqnarray*}
Also, let $K > 0$ be fixed. Then there exists $C_4$ only dependent on $M,\, \psi,\, \delta,\, d,\, t,\, ||f||_\infty,\, K$ such that for all $i$ and $\eta \in \mathcal{H}_i$ that satisfies $d(x,\eta) \leq K2^{-i}$, we have:
\begin{equation*}
|\beta_{i\eta}| \leq C_4 2^{-i(2t+d-1)/2}.
\end{equation*}
\end{proposition}

We also state the minimax optimal rate for density estimation for the spaces we have defined:
\begin{theorem} \label{minimax} Let $t>0,\,\eta\in S^{d-1}$ and let $f:S^{d-1} \to [0,\infty)$ be a probability density. Further, let $X = \{X_1,\,X_2,\ldots,\,X_n\}$ be an independent and identically distributed sample of $n$ values from $f$ and $\mathcal{F}_n =\{\hat{f}:S_{d-1}^n \to [0,\infty)\}$. Then there exists $c,\,\delta,\,M>0$  such that:
\begin{equation*}
\liminf_n \inf_{\hat{f}\in\mathcal{F}_n} \sup_{f\in C_{M,\delta}^t(\eta)} \mathbb{E}_f \Big|n^{\frac{t}{2t+d-1}} (\hat{f}(X_1, X_2,\ldots,X_n) - f(\eta))\Big| \geq c
\end{equation*}
\end{theorem}

\section{Localised Density Estimation by Needlets}

We first show that the linear needlet estimator $\hat{f}^L_J$:
\begin{equation}
\label{ll}
\hat{f}_J^L = \frac{1}{\omega_{d-1}} + \sum_{i=0}^{J-1} \sum_{\eta \in \mathcal{H}_i} \hat{\beta}_{i\eta} \psi_{i\eta}, ~~~\textit{where}~~~ \hat{\beta}_{i\eta} = \frac{1}{n}\sum_{k=1}^n \psi_{i\eta}(X_k)
\end{equation}
is minimax-optimal over the regularity space of functions $C_{M,\delta}^t(x)$ if we are allowed to pick $J$ as a function of the regularity $t$. We will need the following proposition to show that the estimated coefficients are not too far from the true values.
\begin{proposition}
\label{variance}
Let $x \in S^{d-1}$, $f: S^{d-1} \rightarrow \mathbb{R}$ be bounded. Further, let, $X = \{X_1,\,X_2,\ldots,\,X_n\}$ be an independent and identically distributed sample of $n$ values from some probability density function $f$, and $\hat{f}_j^L$ as defined in Equation~\ref{ll}. Then there exists a constant $C_5$ such that for all $i,\,j$:
\begin{eqnarray*}
\mathbb{E}(|\hat{\beta}_{i\eta} - \beta_{i\eta}|) \leq \sqrt{\frac{||f||_\infty}{n}} \\
\mathbb{E}|\hat{f}_j^L(x) - A_j(f)(x)| \leq C_5 \sqrt{||f||_\infty} 2^{j(d-1)/2} n^{-1/2}
\end{eqnarray*}
Furthermore, let $2^{i(d-1)} \leq \frac{n}{\log n}$, $v>0$ and $\kappa_{(v)} = \max(14v/(3\sqrt{\omega_{d-1}}),||f||_\infty \sqrt{\omega_{d-1}})$. Then: 
\begin{equation*}
\mathbb{P}\left(|\hat{\beta}_{i\eta} - \beta_{i\eta}| > \kappa \sqrt{\frac{\log n}{n}}\right) \leq 2 n^{- v}
\end{equation*}
\end{proposition}

Using this result together with Proposition \ref{bias} it is easy to see that $2^J \simeq n^{\frac{1}{2t+d-1}}$ balances the bias and variance terms, and that the resulting local error of estimation satisfies $$\mathbb{E}|\hat{f}_J^L(x) - f(x)| =O(n^{-t/(2t+d-1)})~~~\text{if}~~~ f \in C_{M, \delta}^t(x).$$ This corresponds to the local minimax rate of estimation at $x \in S^{d-1}$. This needlet estimator, although minimax optimal, requires the knowledge of $t$, which is typically not available. Circumventing this knowledge is the subject of this paper.

Following Baldi et. al in \cite{Baldi}, we define the hard thresholded needlet estimator as
\begin{equation}
\label{gotcha}
\hat{f}_{J,\kappa}^{HT} = \frac{1}{\omega_{d-1}} + \sum_{i=0}^{J-1} \sum_{\eta \in \mathcal{H}_i} \hat{\beta}_{i\eta} \psi_{i\eta} 1_{\hat{\beta}_{i\eta} \geq \kappa \sqrt{\frac{\log n}{n}}}
\end{equation}

A first main result is to show that the thresholded needlet estimator is locally minimax optimal within log-factors. Note that the only unknown quantity required for the construction of the thresholded estimator is $||f||_\infty$, which can be replaced by $||\hat{f}_J||_\infty$ in practice, see for instance \cite{ugh}. 

\begin{theorem}
\label{crazycheese}
Let $n \geq 2$ and $X = \{X_1,\,X_2,\ldots,\,X_n\}$ be an independent and identically distributed sample of $n$ values from some probability density function $f: S^{d-1} \rightarrow \mathbb{R}, f \in C^t_{M,\delta}(x), x \in S^{d-1}$. Let $\hat{f}_{J,\kappa}^{HT}$ be defined as in equation~\ref{gotcha}. If $2^{J(d-1)} \simeq \frac{n}{\log n}$ and $\kappa = 2\max(14/3\sqrt{\omega_{d-1}},||f||_\infty \sqrt{\omega_{d-1}})$, then there exists $C = C(M, \psi, \delta, d, t, ||f||_\infty)$ such that:
\begin{equation*}
\sup_{f \in C^t_{M,\delta}(x)} \mathbb{E} |\hat{f}_{J,\kappa}^{HT}(x) - f(x)| \leq  C \left(\frac{n}{\log n}\right)^{-\frac{t}{2t+d-1}}
\end{equation*}
\end{theorem}

We now aim to utilise the local adaptation property of the thresholded needlet estimator to construct confidence intervals. This means that given the sample $X_1, X_2,\ldots, X_n$ we want to choose a data-driven $\hat{\sigma}_{J,\alpha}(x)$ such that:
\begin{enumerate}
\item The confidence interval $[\hat{f}_{J,\kappa}^{HT}(x) - 1.01\hat{\sigma}_{J,\alpha}(x), \hat{f}_{J,\tau}^{HT}(x) + 1.01\hat{\sigma}_{J,\alpha}(x)]$ is asymptotically honest with level $\alpha$ for every $x \in S^{d-1}$:
\begin{equation*}
\liminf_n \inf_{t > 0} \inf_{f \in C^t_{M,\delta}(x)} \mathbb{P}_{f^n} \big(f(x) \in [\hat{f}_{J,\tau}^{HT}(x) - 1.01\hat{\sigma}_{J,\alpha}(x), \hat{f}_{J,\tau}^{HT}(x) + 1.01\hat{\sigma}_{J,\alpha}(x)]\big) \geq 1-\alpha
\end{equation*}
\item The expected size of the confidence interval shrinks at the right rate in $n$, up to $\log n$ terms. 
\end{enumerate}

However, this will be a fruitless task, as the results of Low in \cite{Low} indicate that no honest confidence interval C(x) can adapt to the local smoothness of $f$. This is due to pathological functions which masquerade as having a higher H\"{o}lder exponent than they actually do. With this in mind, the method indicated in \cite{nickl, winner, more} is to choose subsets $\bar{C}^t_{\delta,M}(x) \subset C^t_{\delta,M}(x)$ to cut out these pathologies so that the smoothness parameter $t$ becomes identified. In the case as treated by Kerkyacharian, Nickl and Picard in \cite{nickl2}, we know that if $f \in C^t(S^{d-1})$ were globally t-H\"{o}lder, then $||A_j(f) - f||_\infty \leq D_2 2^{-jt}$, hence it is natural to form $\bar{C}^t \subset C^t$ by admitting only functions $f \in C^t(S^{d-1})$ which satisfy 
\begin{equation*}
D_1 2^{-jt} \leq ||A_j(f) - f||_\infty \leq D_2 2^{-jt},
\end{equation*}
and there it was shown that 'typical' H\"{o}lder functions in $C^t(S^{d-1})$ on the sphere satisfy this condition. There are more compelling results in wavelet theory - Gin\'e and Nickl in \cite{nickl} showed that quasi-every function in $C^t(\mathbb{R})$ satisfies this lower bound condition, whilst cutting out masquerading pathologies. However, this condition is a global property and is thus not suitable for our purposes. We need to find a local analogue of this self-similarity property.

First, recall that at each level $i$, the needlets have a maximum height of order $2^{(d-1)/2}$. Since we are only interested in the neighbourhood of a point $x$, we only consider needlets which have this maximum height. We thus conclude that the centers $\eta$ of these needlets must be close to $x$ by Equation~\ref{psiimpt}: $$\psi_{i\eta}(x) 2^{-i(d-1)/2} \rightarrow 0~~~\text{as}~~~~2^id(x,\eta) \rightarrow 0,$$ hence $d(x, \eta) < A2^{-i}$ for some $A>0$. Now, by Proposition \ref{bias} above $$|\beta_{i \eta}| \le C2^{-i(2t+d-1)/2}$$ holds for those $\eta \in \mathcal H_i$. The following lower bound condition thus becomes natural:
\begin{condition}
\label{lala}
Let $f \in C^t_{M,\delta}(x)$. There exists $A,B >0$ and a sequence $0 < \rho_n < 1$ such that:
$$\max_{\eta \in \mathcal{H}_i} |\beta_{i\eta}| 1_{\psi_{i\eta}(x) \geq A 2^{i(d-1)/2}} \geq B 2^{-i(2t+d-1)/2}$$
where $2^{i(2t+d-1)} \simeq n \rho_n$.
\end{condition}
We note that from Lemma~\ref{lastbutnotleast} below, $1_{\psi_{i\eta}(x) \geq A 2^{i(d-1)/2}} =1$ for some $A$ if $d(x,\eta) < C_8 2^{-i}$, hence, if the quadrature were dense enough, then there are $\eta$ which satisfy this condition. This sequence $\rho_n$ is the needlet-analogue of the condition used in \cite{more}, and it measures the loss of adaptation. We should note that for example, each component of the sum $$f(x) = \frac{1}{\omega_d} + \sum_{j\geq1} a_j Z^{2^j}(\eta,x)$$ is orthogonal to all but one of the needlet levels; hence the $a_j$ can be chosen such that Condition \ref{lala} holds for any sequence $\rho_n$.

Proposition~\ref{bias} suggests to take $\hat{\beta}_{i\eta} \pm \kappa_{(v)} \frac{\log n}{\sqrt{n}}$ as a confidence interval for each non-thresholded needlet coefficient $\beta_{i \eta}$. Our second main result is now the following theorem.

\begin{theorem}
\label{FINISHHIM}Let $n \geq 2$, and let $X = \{X_1,\,X_2,\ldots,\,X_n\}$ be an independent and identically distributed sample from some probability density function $f: S^{d-1} \rightarrow \mathbb{R}$ and let:
\begin{eqnarray*}
\hat{f}_{J,{2\kappa_1}}^{HT}(x) = \frac{1}{\omega_d} + \sum_{i=0}^{J-1} \sum_\eta \hat{\beta}_{i\eta} 1_{|\hat{\beta}_{i\eta}| \geq 2\kappa_1 \sqrt{\frac{\log n}{n}}} \psi_{i\eta}(x) \notag \\
\hat{\sigma}_{J,\alpha}(x) = \sum_{i=0}^{J-1} \sum_\eta \left|\kappa_{w} \gamma_n \sqrt{\frac{\log n}{n}}  1_{|\hat{\beta}_{i\eta}| \geq 2\kappa_1 \sqrt{\frac{\log n}{n}}} \psi_{i\eta}(x)\right|
\end{eqnarray*}

with $\kappa_{(v)} = \max(14v/(3\sqrt{\omega_{d-1}}),||f||_\infty \sqrt{\omega_{d-1}})$, $w$ such that $n^{-w} = \frac{\alpha}{n}$ and $J$ such that $2^{J(d-1)} \simeq \frac{n}{\log n}$. Then if $f \in C^t_{M,\delta}(x)$ satisfies Condition~\ref{lala} such that $\gamma_n (\rho_n \log n)^{(d-1)/(2(2t+d-1))}$ diverges and $\rho_n \log n$ converges to 0, then there exists $C$ dependent on $M,\,\psi,\,\delta,\, d,\, t,\, ||f||_\infty,\,\alpha$ such that:
\begin{eqnarray*}
\mathbb{E} \hat{\sigma}_{J,\alpha}(x) \leq C \Big(\frac{n}{\log n}\Big)^{-\frac{t}{2t+d-1}}  \gamma_n \\
\liminf_n \mathbb{P}_{f^n} \big(f(x) \in [\hat{f}_{J,2\kappa_1}^{HT}(x) - 1.01\hat{\sigma}_{J,\alpha}(x), \hat{f}_{J,2\kappa_1}^{HT}(x) + 1.01\hat{\sigma}_{J,\alpha}(x)]\big) \geq 1-\alpha
\end{eqnarray*}
\end{theorem}

\begin{remark} Inspection of the proof of Theorem~\ref{FINISHHIM} shows that the confidence interval in Theorem 2 is asymptotically honest over $f \in \bigcup_{t \in [r,R]} \{\{C_{M,\delta}^t(x)\} \bigcap~\{\text{Condition \ref{lala}}(A, B, \rho_n)\}\}$. \end{remark}
\begin{remark} For simplicity, we have chosen $w$ to accommodate all $O(n)$ needlet coefficients. In practice, we suggest that $w$ be chosen such that $n^{-w} = \frac{\alpha}{\#}$ where $\#$ is the number of non-thresholded coefficients, which will decrease the size of the interval without affecting the theoretical results. \end{remark}
\begin{remark} To ensure that the 'large' needlet in Condition~\ref{lala} is not truncated, we require that $|\beta_{i\eta} | \geq B2^{-i(2t+d-1)/2} = \frac{B}{\sqrt{n\rho_n}}$ be larger than $2 \kappa_1 \sqrt{\frac{\log n}{n}}$. Hence, $\log n \rho_n$ has to converge to 0. This is the loss of adaptation in hard thresholding. Nonetheless, if we assume the function is exactly self-similar, we can pick $\rho_n = (\log n)^{-2}$, say, and allows us to pick $\gamma_n = \log n$ gives us a confidence interval whose expected size is minimax-optimal up to $\log n$ terms. \end{remark}
\begin{remark} For practical implementation, the only unknown quantity we require is $||f||_\infty$ This may be replaced by $||\hat{f}_J||_\infty$ in practice, see for instance the proof of Theorem 2 in \cite{ugh}. However it must be stressed that although the constants given here are sufficient for the theoretical results, however practical implementation may require better constants, and this is beyond the scope of this paper. \end{remark}  

\section{Proofs}
Our main focuses here are Theorems~\ref{crazycheese} and~\ref{FINISHHIM} so we will start by proving these. We will assume Lemma~\ref{salamander}, Propositions~\ref{bias} and~\ref{variance} in these sections. We will then prove these three statements. Finally, we will give proofs of the lower bounds for $\psi$ and the size of the confidence intervals.

\subsection{Proof of Theorem~\ref{crazycheese}}

\begin{proof}[Proof of Theorem~\ref{crazycheese}]
Let $2^{J_1} \simeq \left(\frac{n}{\log n}\right)^{\frac{1}{2t+d-1}}$ in such a way that $J \geq J_1$. We have that:
\begin{align*}
\mathbb{E}|\hat{f}_{J,\tau}^{HT}(x) - f(x)| &\leq \mathbb{E}|\hat{f}_{J_1}^L(x) - f(x)| \\ 
&+ \mathbb{E} \left|\sum_{i=0}^{J_1-1} \sum_\eta \hat{\beta}_{i\eta} 1_{|\hat{\beta}_{i\eta}| < \kappa \sqrt{\frac{\log n}{n}}} \psi_{i\eta}(x)\right| \\ 
& + \mathbb{E} \left|\sum_{i=J_1}^{J-1} \sum_\eta \hat{\beta}_{i\eta}  1_{|\hat{\beta}_{i\eta}| \geq \kappa \sqrt{\frac{\log n}{n}}} \psi_{i\eta}(x)\right|.\\
\end{align*}

By Propositions~\ref{bias} and~\ref{variance}, we have:
\begin{align*}
\mathbb{E}|\hat{f}_{J_1}^L(x) - f(x)| &\leq |A_{J_1}(f)(x) - f(x)| + \mathbb{E}|\hat{f}_{J_1}^L(x) - A_{J_1}(f)(x)| \\
&\leq C_2 2^{-J_1t} + C_5 \sqrt{||f||_\infty}2^{J_1(d-1)/2} n^{-1/2} \\
&\leq (C_2 + C_5\sqrt{||f||_\infty}) \left(\frac{n}{\log n}\right)^{-\frac{t}{2t+d-1}}
\end{align*}

Now we bound the second term using Lemma~\ref{salamander}:

\begin{align*}
\mathbb{E} \left|\sum_{i=0}^{J_1-1} \sum_\eta \hat{\beta}_{i\eta} 1_{|\hat{\beta}_{i\eta}| < \kappa \sqrt{\frac{\log n}{n}}} \psi_{i\eta}(x)\right| &\leq \sum_{i=0}^{J_1-1} \sum_\eta \left|\kappa \sqrt{\frac{\log n}{n}} \psi_{i\eta}(x)\right| \\
&\leq \sum_{i=0}^{J_1-1} C_1 \sqrt{\frac{\log n}{n}} 2^{i(d-1)/2} \\
&\leq C_1 \sqrt{\frac{\log n}{n}} \left(\frac{n}{\log n}\right)^{-\frac{t}{2t+d-1}}.
\end{align*}

Now, we deal with the other term. We have:
\begin{align*}
&\mathbb{E} \left|\sum_{i=J_1}^{J-1} \sum_\eta \hat{\beta}_{i\eta}  1_{|\hat{\beta}_{i\eta}| \geq \kappa \sqrt{\frac{\log n}{n}}} \psi_{i\eta}(x)\right| \\ 
&\leq \mathbb{E} \sum_{i=J_1}^{J-1} \sum_\eta \left|\beta_{i\eta}  1_{|\hat{\beta}_{i\eta}| \geq \kappa \sqrt{\frac{\log n}{n}}} \psi_{i\eta}(x)\right| \\ 
&+ \mathbb{E} \sum_{i=J_1}^{J-1} \sum_\eta \left|(\hat{\beta}_{i\eta}-{\beta}_{i\eta})  1_{|\beta_{i\eta}| \geq \frac{\kappa}{2} \sqrt{\frac{\log n}{n}}} \psi_{i\eta}(x)\right| \\ 
&+ \mathbb{E} \sum_{i=J_1}^{J-1} \sum_\eta \left|(\hat{\beta}_{i\eta}-{\beta}_{i\eta})  1_{|\hat{\beta}_{i\eta}| \geq \kappa \sqrt{\frac{\log n}{n}}, |\beta_{i\eta}| \leq \frac{\kappa}{2} \sqrt{\frac{\log n}{n}}} \psi_{i\eta}(x)\right| 
\end{align*} 

First, we use Proposition~\ref{bias}.
\begin{align*}
\mathbb{E} \sum_{i=J_1}^{J-1} \sum_\eta \left|\beta_{i\eta}  1_{|\hat{\beta}_{i\eta}| \geq \kappa \sqrt{\frac{\log n}{n}}} \psi_{i\eta}(x)\right| 
&\leq \sum_{i=J_1}^{J-1} \sum_\eta \left|\beta_{i\eta} \psi_{i\eta}(x)\right| \\
&\leq \sum_{i=J_1}^{J-1} C_3 2^{-it} \\ 
&\simeq C_3 \left(\frac{n}{\log n}\right)^{-\frac{t}{2t+d-1}}.
\end{align*}

We then obtain the bound for the second term using Propositions~\ref{bias} and~\ref{variance}:
\begin{align*}
\mathbb{E} \sum_{i=J}^{J-1} \sum_\eta \left|(\hat{\beta}_{i\eta}-{\beta}_{i\eta})  1_{|\beta_{i\eta}| \geq \frac{\kappa}{2} \sqrt{\frac{\log n}{n}}} \psi_{i\eta}(x)\right| 
&\leq \sum_{i=J_1}^{J-1} \sum_\eta \frac{|\beta_{i\eta} \psi_{i\eta}(x)|}{\frac{\kappa}{2} \sqrt{\frac{\log n}{n}}} \mathbb{E} |(\hat{\beta}_{i\eta} - \beta_{i\eta})| \\
&\leq \sum_{i=J_1}^{J-1} C_3 2^{-it} \frac{2}{\kappa} \frac{\sqrt{n}}{\sqrt{\log n}} \frac{\sqrt{||f||_\infty}}{\sqrt{n}} \\
&\leq C_3 (||f||_\infty \omega_{d-1})^{-1/2} n^{-\frac{t}{2t+d-1}} \left(\log n\right)^{-\frac{(d-1)/2}{2t+d-1}}
\end{align*}

Using the Cauchy-Schwarz inequality and Lemma \ref{salamander} we deduce:
\begin{align*}&\mathbb{E} \sum_{i=J_1}^{J-1} \sum_\eta \left|(\hat{\beta}_{i\eta}-{\beta}_{i\eta})  1_{|\hat{\beta}_{i\eta}| \geq \kappa \sqrt{\frac{\log n}{n}}, |\beta_{i\eta}| \leq \frac{\kappa}{2} \sqrt{\frac{\log n}{n}}} \psi_{i\eta}(x)\right|  \\ 
&\leq \mathbb{E} \sum_{i=J_1}^{J-1} \sum_\eta \left|(\hat{\beta}_{i\eta}  - \beta_{i\eta}) \psi_{i\eta}(x) 1_{|\hat{\beta}_{i\eta} - \beta_{i\eta}| > \frac{\kappa}{2} \sqrt{\frac{\log n}{n}}} \right| \\
&\leq \sum_{i=J_1}^{J-1} \sum_\eta |\psi_{i\eta}(x)| \sqrt{\mathbb{E} \left|\hat{\beta}_{i\eta}  - \beta_{i\eta}\right|^2 \mathbb{P}(|\hat{\beta}_{i\eta} - \beta_{i\eta}| > \frac{\kappa}{2} \sqrt{\frac{\log n}{n}})} \\
&\leq \sum_{i=J_1}^{J-1} \sum_\eta |\psi_{i\eta}(x)|  \sqrt{\frac{2}{n} \frac{||f||_\infty}{n}} \\
&\leq \sum_{i=J_1}^{J-1} C_1 \sqrt{2 ||f||\infty} 2^{i(d-1)/2} n^{-1}
\simeq C_1 \sqrt{2 ||f||_\infty} \frac{1}{\sqrt{n \log n}}.
\end{align*}
So by combining these inequalities, we have the result.
\end{proof}

\subsection{Proof of Theorem~\ref{FINISHHIM}}
\begin{proof}[Proof of Theorem~\ref{FINISHHIM}]
We start off by showing that $\hat{\sigma}_{J,\alpha}$ is of the right size. We first note that $w$ is bounded above by $1 + \log_2 \alpha$, and hence $\kappa_w$ is also bounded above, so we will omit it from what follows. We will also drop the $\gamma_n$ which appears on both sides. We split sum into three bits using the triangle inequality:
\begin{align*}
\mathbb{E} \sum_{i=0}^{J-1} \sum_\eta \left|\sqrt{\frac{\log n}{n}}  1_{|\hat{\beta}_{i\eta}| \geq \kappa \sqrt{\frac{\log n}{n}}} \psi_{i\eta}(x)\right| \leq &\sum_{i=0}^{J_1-1} \sum_\eta \left|\sqrt{\frac{\log n}{n}} \psi_{i\eta}(x)\right| \\ 
&+ \mathbb{E} \sum_{i=J_1}^{J-1} \sum_\eta \left|\sqrt{\frac{\log n}{n}}  1_{|\hat{\beta}_{i\eta}| \geq 2\kappa_1 \sqrt{\frac{\log n}{n}}, |\beta_{i\eta}| \leq \kappa_1 \sqrt{\frac{\log n}{n}}} \psi_{i\eta}(x)\right| \\ 
&+ \sum_{i=J_1}^{J-1} \sum_\eta \left|\sqrt{\frac{\log n}{n}}  1_{|\beta_{i\eta}| \geq \kappa_1 \sqrt{\frac{\log n}{n}}} \psi_{i\eta}(x)\right|
\end{align*}
where $2^{J_1} = \Big(\frac{n}{\log n}\Big)^{\frac{1}{2t+d-1}}$. We call the terms $D+E+F$.

We bound the first term using Lemma~\ref{salamander}:
\begin{equation*}
D \leq \sum_{i=0}^{J_1-1} C_1 \frac{2^{i(d-1)/2} \sqrt{\log n}}{\sqrt{n}} \leq C_1 \frac{2^{J_1(d-1)/2} \sqrt{\log n}}{\sqrt{n}} \simeq C_1 \Big(\frac{n}{\log n}\Big)^{-\frac{t}{2t+d-1}}.
\end{equation*}

We bound $E$ using Proposition~\ref{variance} and Lemma~\ref{salamander}:
\begin{align*}
E &\leq \sum_{i=J_1}^{J-1} \sqrt{\frac{\log n}{n}} \sum_\eta |\psi_{i\eta}(x)| \mathbb{P}\left(|\hat{\beta}_{i\eta} - \beta_{i\eta}| \geq \kappa_1 \sqrt{\frac{\log n}{n}}\right) \\
&\leq \sum_{i=J_1}^{J-1} \sqrt{\frac{\log n}{n}} C_1 2^{i(d-1)/2} 2n^{-1} \\ 
&\leq 2C_1 n^{-1} \Big(\frac{n}{\log n}\Big)^{-\frac{t}{2t+d-1}}.
\end{align*}

We bound $F$ using Proposition~\ref{bias}:
\begin{align*}
F \leq \sum_{i=J_1}^{J-1} \sum_\eta \left|\sqrt{\frac{\log n}{n}}  \frac{\beta_{i\eta}}{\kappa_1 \sqrt{\frac{\log n}{n}}} \psi_{i\eta}(x)\right| &\leq \frac{3 \sqrt{\omega_{d-1}}} {14} \sum_{i=J_1}^{J-1} \sum_\eta \left|\beta_{i\eta} \psi_{i\eta}(x)\right| \\ 
&\leq \frac{3 C_3 \sqrt{\omega_{d-1}}} {14} \Big(\frac{n}{\log n}\Big)^{-\frac{t}{2t+d-1}}.
\end{align*}
Summing these inequalities, we have the result. 

Now we show asymptotic coverage. We aim to show that for all $f \in C^t_{\delta,M}(x)$ satisfying Condition~\ref{lala}:
\begin{equation*}
\liminf_n \mathbb{P}_{f^n} \big(|f(x) - \hat{f}_{J,2\kappa_1}^{HT}(x)| \leq 1.01\hat{\sigma}_{J,\alpha}(x)\big) \geq 1-\alpha
\end{equation*}

Define the event $E_n$ as follows:
\begin{equation*}
E_n=\left\{\forall i\leq J-1,\,\eta,\, \left|\beta_{i\eta} - \hat{\beta}_{i\eta}\right| \leq \left|\kappa_{w}\sqrt{\frac{\log n}{n}} \right|\right\}
\end{equation*}

We also split $f$ into two terms $f_1+f_2$:
\begin{align*}
&f = \left(\frac{1}{\omega_{d-1}} + \sum_{i=0}^{J-1} \sum_{\eta \in \mathcal{H}_i} 1_{|\hat{\beta}_{i\eta}| > 2\kappa_1 \sqrt{\frac{\log n}{n}}} \beta_{i\eta} \psi_{i\eta}(x)\right) \\ &+ \left(\sum_{i=0}^{J-1} \sum_{\eta \in \mathcal{H}_i} 1_{|\hat{\beta}_{i\eta}| < 2\kappa_1 \sqrt{\frac{\log n}{n}}} \beta_{i\eta} \psi_{i\eta}(x) + \sum_{i=J}^{\infty} \sum_{\eta \in \mathcal{H}_i}  \beta_{i\eta} \psi_{i\eta}(x)\right)
\end{align*}

Then the following are true:
\begin{enumerate}
\item \label{persimmon} There exists $N$ such that for all $n > N$, $\mathbb{P}(E_n) \geq 1-\alpha$
\item \label{alohomora} Under $E_n$, $|f_1 - \hat{f}_{J,2\kappa_1}^{HT}(x)| \leq \hat{\sigma}_{J,\alpha} (x)$
\item \label{wingardium} Under $E_n$, there exists a constant $C$ dependent on $M,\,\psi,\,\delta,\,d,\,t,\,||f||_\infty,\,\alpha$ such that $|f_2| \leq  C\left(\frac{n}{\log n}\right)^{-\frac{t}{2t+d-1}}$
\item \label{leviosa} Under $E_n$, there exists constants $N$ and $C'$ dependent on $||f||_\infty$ and $A,\,B,\,\rho_n$ from Condition~\ref{lala} such that for all $n > N$, $\hat{\sigma}_{J,\alpha} (x) \geq C' \gamma_n \sqrt{\rho_n (\log n)^{(d-1)/(2t+d-1)}} \left(\frac{n}{\log n}\right)^{-\frac{t}{2t+d-1}}$
\end{enumerate}

\begin{proof}[Proof of part~\ref{persimmon}]
By Proposition~\ref{variance}, we know that 
\begin{equation*}
\mathbb{P}\left(\left|\beta_{i\eta} - \hat{\beta}_{i\eta}\right| \geq \left|\kappa_{w} \sqrt{\frac{\log n}{n}} \right|\right) \leq \alpha n^{-1}
\end{equation*}
Now, $|\mathcal{H}_i| = C 2^{i(d-1)}$, so $\sum_{i=0}^{J-1} |\mathcal{H}_i| \leq C 2^{J(d-1)}$, and thus we have:
\begin{equation*}
\mathbb{P}\left(\forall i\leq J-1,\,\eta_i,\, \left|\beta_{i\eta} - \hat{\beta}_{i\eta}\right| \leq \left|\kappa_{w} \sqrt{\frac{\log n}{n}} \right|\right) \leq C 2^{J(d-1)} \alpha n^{-1} \simeq C \frac{\alpha}{\log n}
\end{equation*}
Taking $\log N = C$ yields the result. 
\end{proof}

\begin{proof}[Proof of part~\ref{alohomora}]
We can rewrite $|f_1 - \hat{f}_{J,{2\kappa_1}}^{HT}(x)|$ as
\begin{equation*}
\left|f_1 - \hat{f}_{J,{2\kappa_1}}^{HT}(x)\right| \le \sum_{i=0}^{J} \sum_{\eta \in \mathcal{H}_i} 1_{|\hat{\beta}_{i\eta}| > 2\kappa_1 \sqrt{\frac{\log n}{n}}} \left|\beta_{i\eta} - \hat{\beta}_{i\eta}\right| \left|\psi_{i\eta}(x)\right|
\end{equation*}
This is clearly smaller than $\hat{\sigma}_{J,\alpha}(x) = \sum_{i=0}^{J-1} \sum_{\eta \in \mathcal{H}_i} 1_{|\hat{\beta}_{i\eta}| > 2\kappa_1 \sqrt{\frac{\log n}{n}}} \left|\kappa_{w} \frac{\log n}{\sqrt{n}} \psi_{i\eta}(x)\right|$ on the event $E_n$
\end{proof}

\begin{proof}[Proof of part~\ref{wingardium}] Let $2^{J_1} \simeq \left(\frac{n}{\log n}\right)^{\frac{1}{2t+d-1}}$ such that $J_1 \leq J$. On $E_n$, we have:
\begin{align*}
|f_2| &\leq \left|\sum_{i=0}^{J_1-1} \sum_{\eta \in \mathcal{H}_i} 1_{|\beta_{i\eta}| < \left(2\kappa_1 + \kappa_{w}\right)\sqrt{\frac{\log n}{n}}} \beta_{i\eta} \psi_{i\eta}(x)\right| \\ 
&+ \left|\sum_{i=J_1}^{J-1} \sum_{\eta \in \mathcal{H}_i} 1_{|\beta_{i\eta}| < \left(2\kappa_1 + \kappa_{w}\right)\sqrt{\frac{\log n}{n}}} \beta_{i\eta} \psi_{i\eta}(x)\right| \\
&+ \left|\sum_{i=J}^{\infty} \sum_{\eta \in \mathcal{H}_i}  \beta_{i\eta} \psi_{i\eta}(x)\right|
\end{align*}

We bound the first term by Lemma~\ref{salamander}:
\begin{equation*}
\sum_{i=0}^{J_1} \sum_\eta \left(2\kappa_1 + \kappa_{w}\right)\sqrt{\frac{\log n}{n}} |\psi_{i\eta}(x)|  \leq  C_1 \left(2\kappa_1 + \kappa_{w}\right) \left(\frac{n}{\log n}\right)^{-\frac{t}{2t+d-1}}
\end{equation*}

We bound the second term using Proposition~\ref{bias} by:
\begin{equation*}
\sum_{i=J_1 + 1}^J \sum_\eta  |\beta_{i\eta}| |\psi_{i\eta}(x)| \leq C_3 \left(\frac{n}{\log n}\right)^{-\frac{t}{2t+d-1}}
\end{equation*} 

Finally, by Proposition~\ref{bias}, the last term is bounded by $C_2 2^{-Jt} = C_2 n^{-\frac{t}{d-1}}$. Summing the inequalities gives us the result.
\end{proof}

\begin{proof}[Proof of part~\ref{leviosa}]
Let $i$ and $\eta$ be the arguments in Condition~\ref{lala}, that is:
\begin{equation*}
|\beta_{i\eta}| 1_{\psi_{i\eta}(x) \geq A 2^{i(d-1)/2}} \geq B 2^{-i(2t+d-1)/2}
\end{equation*}
Now, from the definition of $i$, $|\beta_{i\eta}| \geq B (n \rho_n)^{-1/2} $. Since $\rho_n \log n$ converges to 0, there exists $N'$ depending only on $B,\, \rho_n,\,||f||_\infty$ such that for all $n > N'$, $B (n \rho_n')^{-1/2} \geq (2 \kappa_1 + \kappa_w)\sqrt{n/\log n}$ (recalling that $\kappa_w$ is bounded above by a fixed constant). Thus:
\begin{align*}
\hat{\sigma}_{J,\alpha} (x) &\geq \left|\kappa_{w} \gamma_n \sqrt{\frac{\log n}{n}}  A 2^{i(d-1)}\right| \\
&\geq \left|A \kappa_{w} \gamma_n (\rho_n \log n)^{(d-1)/(2(2t+d-1))} \left(\frac{n}{\log n}\right)^{- \frac{t}{2t+d-1}}\right|
\end{align*}
which together with $\kappa_w \geq 1$ gives us the result.
\end{proof}
\\
\\
We can choose $N$ large enough such that under $E_n$, $0.01\hat{\sigma}_{J,\alpha} (x) \geq |f_2|.$ Thus, by the triangle inequality, under $E_n$, $0.01\hat{\sigma}_{J,\alpha} (x) \geq |f_1 - \hat{f}_{J,{2\kappa_1}}^{HT}(x)|$, and using the triangle inequality again with part~\ref{alohomora} gives asymptotic coverage. \end{proof}

\subsection{Upper bounds on the $L^2$ and $L^\infty$ norms of $\psi_{i\eta}$}
We will need the following bounds on the $L^2$ and $L^\infty$ norms of $\psi_{i\eta}$ to prove the propositions and lemma.

\begin{lemma}
\label{quad}
For all $\eta \in \mathcal{H}_i$, $||\psi_{i\eta}||_2 \leq 1$ and $||\psi_{i\eta}||_\infty \leq 2\omega_{d-1}^{-1/2}2^{i(d-1)/2}$
\end{lemma}

\begin{proof}Since $x \rightarrow (\sum_{k\leq2^{i+1}} Z^k(\eta,x))^2$ is a polynomial of degree $2^{i+2}$, by quadrature,
\begin{equation*}
\int_{S^{d-1}}\left(\sum_{k\leq2^{i+1}} Z^k(\eta,x)\right)^2\,dx = \sum_{\eta \in \mathcal{H}_i} \lambda_{\eta} \left(\sum_{k\leq2^{i+1}} Z^k(\eta,\eta)\right)^2 \geq \lambda_{\eta} \left(\sum_{k\leq2^{i+1}} Z^k(\eta,\eta)\right)^2
\end{equation*}
However, recalling equation~\ref{ort}:
\begin{equation*}
\int_{S^{d-1}}\left(\sum_{k\leq2^{i+1}} Z^k(\eta,x)\right)^2\,dx = \sum_{k\leq2^{i+1}} \int_{S^{d-1}} Z^k(\eta,x)^2\,dx = \sum_{k\leq2^{i+1}} Z^k(\eta,\eta)
\end{equation*}
Thus $\lambda_{\eta} \leq (\sum_{k\leq2^{i+1}} Z^k(\eta,\eta))^{-1}$. Recalling $\psi_{i\eta}(x) = \sqrt{\lambda_\eta} \sum_{2^{i-1}\leq k \leq 2^{i+1}} \sqrt{b\Big(\frac{k}{2^i}\Big)} Z^k(\eta,x)$, since $Z_\eta$ has a maximum of $\left(\binom{d+k-1}{d-1} - \binom{d+k-3}{d-1}\right)\omega_{d-1}^{-1}$ at $\eta$, we have:
\begin{eqnarray*}
||\psi_{i\eta}||_2^2 \leq \int_{S^{d-1}} \lambda_\eta \left(\sum_{2^{i-1}\leq k \leq 2^{i+1}} Z^k(\eta,x)\right)^2\,dx \leq 1 \\
||\psi_{i\eta}||_\infty \leq \sqrt{\lambda_\eta} \sum_{2^{i-1}\leq k \leq 2^{i+1}} Z^k(\eta,\eta) \leq \sqrt{\sum_{k \leq 2^{i+1}} Z^k(\eta,\eta)}
\end{eqnarray*}

Now, by a telescoping sum:
\begin{align*}
\sum_{k\leq2^{i+1}} Z^k(\eta,\eta) &= \left(\binom{d+2^{i+1}-1}{d-1} + \binom{d+2^{i+1}-2}{d-1}\right)\omega_{d-1}^{-1} \\ &\leq \frac{2}{(d-1)!} 2^{(i+1)(d-1)} \omega_{d-1}^{-1} \\ &\leq 4\omega_{d-1}^{-1}2^{i(d-1)}
\end{align*}
Thus $||\psi_{i\eta}||_\infty \leq 2\omega_{d-1}^{-1/2}2^{i(d-1)/2}$ \end{proof}

\subsection{Proofs of Lemma~\ref{salamander} and Proposition~\ref{variance}}
Before we prove Lemma~\ref{salamander}, we need Lemma 6 from Baldi et al.\cite{Baldi}:

\begin{lemma}
\label{sumbound}
There exists $C_6$ such that for all $y$ and $i$, $\sum_{\eta \in \mathcal{H}_i} \frac{1}{(1+2^id(y,\eta))^3} \leq C_6$ 
\end{lemma}

\begin{proof}[Proof of Lemma~\ref{salamander}] We use Lemma~\ref{sumbound} together with the bound for $|\psi_{i\eta}(x)|$ in equation~\ref{psiimpt}:
\begin{equation*}
\sum_{\eta \in \mathcal{H}_i} |\psi_{i\eta}(y)| \leq \sum_\eta \frac{c_3 2^{i(d-1)/2}}{(1+2^id(x,\eta))^3} \leq c_3 C_6 2^{i(d-1)/2}
\end{equation*}
which gives the result.
\end{proof}

This leads us on to proving the first two parts of Proposition~\ref{variance}:

\begin{proof}[Proof of the first two parts of Proposition~\ref{variance}]Since $||\psi_{i\eta}||_2 \leq 1$:

\begin{eqnarray*}
\mathbb{E}(|\hat{\beta}_{i\eta} - \beta_{i\eta}|) \leq \sqrt{\mathbb{E}(|\hat{\beta}_{i\eta} - \beta_{i\eta}|^2)} = \sqrt{Var(\hat{\beta}_{i\eta})} \leq \sqrt{\frac{\mathbb{E}(\psi_{i\eta}(X)^2)}{n}} \leq \sqrt{\frac{||f||_\infty}{n}} \\
\mathbb{E}|\hat{f}_J^L(x) - A_J(f)(x)| \leq \sum_{i=0}^{J-1} \sum_\eta \mathbb{E} |\hat{\beta}_{i\eta} - \beta_{i\eta}| |\psi_{i\eta}(x)| 
\leq \frac{C_1\sqrt{||f||_\infty}2^{J(d-1)/2}}{\sqrt{n}}.
\end{eqnarray*}
So in fact, $C_1$ and $C_5$ are the same constant.
\end{proof}

To prove the third part, we need to invoke the well known Bernstein's inequality:

\begin{lemma}[Bernstein's Inequality]
\label{bernstein}
If $Y_1,\,\ldots,\,Y_n$ are mean zero independent and identically distributed random variables taking values in $[-c, c]$ for some constant $0 < c < \infty$, then
\[
\mathbb{P} \left(\left|\sum_{k=1}^n Y_k\right| > u\right) \leq 2 \exp \left(-\frac{u^2}{2n\mathbb{E}Y_1^2 + (2/3)cu}\right)
\]
\end{lemma}

\begin{proof}[Proof of the third part of Proposition~\ref{variance}]Now if $Y_k = \psi_{i\eta}(X_k) - \beta_{i\eta}$, then $\sum_{k=1}^n Y_k = n(\hat{\beta}_{i\eta} - \beta_{i\eta})$ and $\mathbb{E}(Y_i)=0$. We have $\mathbb{E}(Y_k^2) \leq ||f||_\infty$, and thus we just need to find the following bounds:
\begin{eqnarray*}
\mathbb{E}Y_1^2 \leq ||f||_\infty ||\psi_{i\eta}(x)||_2^2 \leq ||f||_\infty \\
c \leq ||Y_k||_\infty \leq 2||\psi_{i\eta}(x)||_\infty \leq 4\omega_{d-1}^{-1/2}2^{i(d-1)/2}
\end{eqnarray*}

Hence, by Bernstein's inequality:

\begin{align*}
\mathbb{P}\left(|\hat{\beta}_{i\eta} - \beta_{i\eta}| > \kappa \sqrt{\frac{\log n}{n}}\right) &\leq \left(\left|\sum_{k=1}^n Y_k\right| > \kappa \sqrt{n \log n}\right) \\
&\leq 2 \exp\left(- \frac{\kappa^2n \log n}{2n||f||_\infty + (8/3)\omega_{d-1}^{-1/2}\kappa \sqrt{n \log n}(2^{i(d-1)/2})}\right) \\
&\leq 2 \exp\left(- \frac{14v\kappa \omega_{d-1}^{-1/2} n \log n/3}{2n \kappa \omega_{d-1}^{-1/2} + (8/3)\kappa \omega_{d-1}^{-1/2} n}\right) \\
&\leq 2 n^{-v}
\end{align*}

using the facts $2^{i(d-1)} \leq \frac{n}{\log n}$ and $\kappa = \max(14v/3\sqrt{\omega_{d-1}},||f||_\infty \sqrt{\omega_{d-1}})$ freely.
\end{proof}

\subsection{Proof of Proposition~\ref{bias}}
We will need the following integral bounds:
\begin{lemma}
\label{dickens}
There exists a constant C such that for all $t>0$:
\begin{eqnarray*}
\int_{S^{d-1}} \frac{d(x,y)^t}{(1+ 2^jd(x,y))^{d+t}}\,dy \leq C 2^{-j(t+d-1)} \\
\int_{S^{d-1}} \frac{1}{(1+ 2^jd(x,y))^{d+t}} 1_{d(x,y) > \delta}\,dy \leq C \delta^{-t} 2^{-j(t+d-1)}
\end{eqnarray*}
\end{lemma}

\begin{proof}
We have an integration formula for zonal functions on $S^{d-1}$, see, e.g., Proposition 9.1.2 in Faraut \cite{faraut}: If $f: S^{d-1} \rightarrow \mathbb{R}^+$ be such that there exists a point $x_0 \in S^{d-1}$ and a function $F: \mathbb{R}^+ \rightarrow \mathbb{R}^+$ such that $f$ has the representation $f(y) = F(d(x_0,y))$. Then if $\theta = \cos^{-1} d(x_0,y)$ is the angle between $x_0$ and $y$, one has
\begin{equation*}
\int_{S^{d-1}} f(y)\,dy = \frac{\Gamma({\frac{d+1}{2}})}{\sqrt{\pi}\Gamma({\frac{d}{2}})} \int_0^\pi F(\theta) \sin^{d-1} \theta\,d\theta \leq C \int_0^\pi F(\theta) \theta^{d-1}\,d\theta. 
\end{equation*}

By this integration formula we have:
\begin{align*}
\int_{S^{d-1}} \frac{d(x,y)^t}{(1+ 2^jd(x,y))^{d+t}}\,dy &\leq C \int_0^\pi \frac{\theta^{d+t-2}}{(1+ 2^j\theta)^{d+t}}\,d\theta \\
&\leq C2^{-j(t+d-1)} \int_0^\infty \frac{u^{d+t-2}}{(1+u)^{d+t}}\,du \\
&\leq C2^{-j(t+d-1)} \int_0^\infty \frac{1}{(1+u)^2}\,du = C2^{-j(t+d-1)}
\end{align*}
using the substitution $u = 2^j\theta$. Similarly:
\begin{align*}
\int_{S^{d-1}} \frac{1}{(1+ 2^jd(x,y))^{d+t-1}} 1_{d(x,y) > \delta}\,dy &\leq C\int_{\delta}^\pi \frac{\theta^{d-2}}{(1+ 2^j\theta)^{d+t-1}}\,d\theta \\
&\leq C2^{-j(d-1)} \int_{2^j\delta}^\infty \frac{u^{d-2}}{(1+ u)^{d+t-1}}\,du \\
&\leq C2^{-j(d-1)} \int_{2^j\delta}^\infty \frac{1}{(1+u)^{t+1}}\,du \leq C2^{-j(t+d-1)} \delta^{-t}
\end{align*}
completing the lemma.\end{proof}

\begin{proof}[Proof of the first part of Proposition \ref{bias}]
We first note that since $A_j(x,y)$ is bounded in equation~\ref{aimpt}, we will only need to consider those $j$ such that $2^j \geq \lfloor t\rfloor$ and ignore the finite number of $j$ which do not fulfill this property. We have a polynomial $P_f$ of degree $\lfloor t \rfloor$ such that $|f(y) - P_f(y)| \leq Md(x,y)^t$ in $B(x,\delta)$. Hence, recalling that $A_j$ reproduces polynomials of degree $\leq 2^j$ we have:
\begin{align*}
|A_j(f)(x) - f(x)| &\leq |A_j(P_f)(x) - P_f(x)| + |A_j(f-P_f)(x)| +|(f-P_f)(x)| \\ &= |A_j(f-P_f)(x)| \\ &= \left|\int_{S^{d-1}} A_j(x,y) (f-P_f)(y)\,dy\right|
\end{align*} because  $(f-P_f)(x) = 0$ and since $j$ is large enough such that $A_j(P_f) = P_f$. We then split this integral into two parts, that over $B(x,\delta)$ and that of its complement, and deal with each part separately. On $B(x,\delta)$, we have a bound on the size of $(f-P_f)(y)$; together with the lemma~\ref{dickens}, we obtain:

\begin{align*}
\left|\int_{B(x,\delta)} A_j(x,y) (f-P_f)(y)\,dy\right| 
&\leq \int_{B(x,\delta)} |A_j(x,y)| |(f-P_f)(y)|\,dy \\
&\leq \int_{S^{d-1}} \frac{c_{d+t} 2^{j(d-1)}}{(1+2^jd(x,y))^{d+t}} M d(x,y)^t\,dy \\
&\leq c_{d+t} M C 2^{-jt}
\end{align*}
Similarly:
\begin{align*}
\left|\int_{B(x,\delta)^c} A_j(x,y) (f-P_f)(y)\,dy\right| 
&\leq (M + ||f||_\infty) \int_{B(x,\delta)^c} |A_j(x,y)| \,dy \\
&\leq (M + ||f||_\infty) \int_{S^{d-1}} \frac{c_{d+t-1} 2^{j(d-1)}}{(1+2^jd(x,y))^{d+t-1}} 1_{d(x,y)>\delta}\,dy \\
&\leq (M + ||f||_\infty) C \delta^{-t} 2^{-jt}
\end{align*}

Summing the inequalities gives us the result. \end{proof}

\begin{proof}[Proof of the second part of Proposition \ref{bias}]
We first note that since $\psi_{i\eta}(y)$ is bounded in equation~\ref{psiimpt}, we will only need to consider those $i$ such that $2^{i-1} \geq \lfloor t\rfloor$ and ignore the finite number of $j$ which do not fulfill this property. we have $\int_{S^{d-1}} \psi_{i\eta}(y)P_f(y)\,dy = 0$ so $\beta_{i\eta} = \int_{S^{d-1}} \psi_{i\eta}(y)(f-P_f)(y)\,dy$. We split the integral into two parts, that of around $B(x,\delta)$ and that of its complement, and deal with each part separately. On $B(x,\delta)$, we have a bound on the size of $(f-P_f)(y)$:
\begin{align*}
&\sum_\eta \left|\int_{B(x,\delta)} \psi_{i\eta}(y) \psi_{i\eta}(x) (f-P_f)(y)\,dy\right| \\
&\leq \sum_\eta \int_{B(x,\delta)} |\psi_{i\eta}(y) \psi_{i\eta}(x)| |(f-P_f)(y)|\,dy \\
&\leq \sum_\eta \int_{S^{d-1}} \frac{c_{d+t} 2^{i(d-1)/2}}{(1+2^id(y,\eta))^{d+t}} \frac{c_{d+t+3} 2^{i(d-1)/2}}{(1+2^id(x,\eta))^{d+t+3}} M d(x,y)^t\,dy \\
&\leq c_{d+t} c_{d+t+3} M \int_{S^{d-1}} \sum_\eta \frac{1}{(1+2^id(x,\eta))^3} \frac{2^{i(d-1)}}{(1+2^id(x,y))^{d+t}}  d(x,y)^t\,dy \\
&\leq c_{d+t} c_{d+t+3} M C_6 C 2^{-it}
\end{align*}
using the fact that $(1+|x|)(1+|y|) \geq 1+|x|+|y|$, $1+2^id(x,\eta) \leq K$ and Lemmas~\ref{sumbound} and~\ref{dickens}. Similarly,
\begin{align*}
&\sum_\eta \left|\int_{B(x,\delta)^c} \psi_{i\eta}(y) \psi_{i\eta}(x) (f-P_f)(y)\,dy\right| \\
&\leq (M + ||f||_\infty) \sum_\eta \int_{B(x,\delta)^c} |\psi_{i\eta}(y) \psi_{i\eta}(x)| \,dy \\
&\leq (M + ||f||_\infty) \sum_\eta \int_{S^{d-1}} \frac{c_{d+t-1} 2^{i(d-1)/2}}{(1+2^id(y,\eta))^{d+t-1}} \frac{c_{d+t+2} 2^{i(d-1)/2} 1_{d(x,y)>\delta}}{(1+2^id(x,\eta))^{d+t+2}}\,dy \\
&\leq (M + ||f||_\infty) c_{d+t-1} c_{d+t+2} \sum_\eta \frac{1}{(1+2^id(x,\eta))^3} \int_{S^{d-1}}  \frac{2^{i(d-1)}}{(1+2^id(x,y))^{d+t-1}} 1_{d(x,y)>\delta}\,dy \\
&\leq (M + ||f||_\infty) c_{d+t-1} c_{d+t+2} C_6 C \delta^{-t} 2^{-it}
\end{align*}

Summing the inequalities gives us the result. \end{proof}

\begin{proof}[Proof of the third part of Proposition \ref{bias}]
We first note that since $\psi_{i\eta}(y)$ is bounded in equation~\ref{psiimpt}, we will only need to consider those $i$ such that $2^{i-1} \geq \lfloor t\rfloor$ and ignore the finite number of $j$ which do not fulfill this property. we have $\int_{S^{d-1}} \psi_{i\eta}(y)P_f(y)\,dy = 0$ so $\beta_{i\eta} = \int_{S^{d-1}} \psi_{i\eta}(y)(f-P_f)(y)\,dy$. We split the integral into two parts, that of around $B(x,\delta)$ and that of its complement, and deal with each part separately. On $B(x,\delta)$, we have a bound on the size of $(f-P_f)(y)$; together with the bound for $\psi_{i\eta}(y)$ in equation~\ref{psiimpt}, we have:
\begin{align*}
\left|\int_{B(x,\delta)} \psi_{i\eta}(y) (f-P_f)(y)\,dy\right|
&\leq c_{t+d} \int_{S^{d-1}} \frac{2^{i(d-1)/2} (K+1)^{d+t}}{{(1+2^id(y,\eta))^{d+t}(1+2^id(x,\eta))^{d+t}}} M d(x,y)^t\,dy \\
&\leq M c_{t+d} (K+1)^{d+t} \int_{S^{d-1}} \frac{2^{i(d-1)/2}}{(1+2^id(y,x))^{d+t}} d(x,y)^t\,dy \\
&\leq M c_{t+d} (K+1)^{d+t} C_6 2^{-i(2t+d-1)/2}
\end{align*}
using the fact that $(1+|x|)(1+|y|) \geq 1+|x|+|y|$, $1+2^id(x,\eta) \leq K$ and Lemmas~\ref{sumbound} and~\ref{dickens}. Similarly,
\begin{align*}
\left|\int_{B(x,\delta)^c} \psi_{i\eta}(y) (f-P_f)(y)\,dy\right| 
&\leq (||f||_\infty + M) \int_{B(x,\delta)^c} |\psi_{i\eta}(y)| \,dy \\
&\leq (||f||_\infty + M) c_{d+t-1} \int_{B(x,\delta)^c} \frac{2^{i(d-1)/2} (K+1)^{d+t-1}}{((1+2^id(y,\eta))(1+2^id(x,\eta))^{d+t-1}} \,dy \\
&\leq (||f||_\infty + M) c_{d+t-1} C \delta^{-t} 2^{-i(2t+d-1)/2}
\end{align*}

Summing the inequalities us gives the result. \end{proof}

For the lower bounds in this paper, we will require the following results:

\subsection{Lower bounds on the size of $\psi_{i\eta}$ near $\eta$}
\begin{lemma}
\label{lastbutnotleast}
Let the needlets be constructed such that $|a(7/4)| \geq c_a > 0$ and $|\lambda_\eta| \geq c_\lambda 2^{-i(d-1)}$ for $\eta \in \mathcal{H}_i$. Then, there exists $C_7, C_8$ depending on $c_a,\, c_\lambda,\, d$ such that for all $x$ such that $d(x,\eta) < C_8 2^{-i}$ and for all $i,\,\eta \in \mathcal{H}_i$:
\begin{equation*}
\psi_{i\eta}(x) \geq C_7 2^{i(d-1)/2}
\end{equation*}
\end{lemma}

\begin{proof}We first show that this is true for $x = \eta$. Recall from the proof of Lemma~\ref{quad} that:
\begin{equation*}
\sum_{k\leq2^i} Z^k(\eta,\eta) = \left(\binom{d+2^i-1}{d-1} + \binom{d+2^i-2}{d-1}\right)\omega_{d-1}^{-1}
\end{equation*}
Hence:
\begin{align*}
\psi_{i\eta}(\eta) &= \sqrt{\lambda_\eta} \sum_{2^{i-1} \leq k \leq 2^{i+1}} \sqrt{b\Big(\frac{k}{2^i}\Big)} Z^k(\eta, \eta) \\
&\geq \sqrt{c_\lambda} 2^{-i(d-1)/2} \sum_{2^i \leq k \leq (7/4)2^i} c_a Z^k(\eta, \eta) \\
&\geq C 2^{-i(d-1)/2} \left(\binom{d+(7/4)2^i-1}{d-1} + \binom{d+(7/4)2^i-2}{d-1} - \binom{d+2^i-1}{d-1} - \binom{d+2^i-2}{d-1}\right) \\
&\geq C 2^{-i(d-1)/2} (3/4) 2^i \left(\binom{d+(7/4)2^i-2}{d-2} + \binom{d+(7/4)2^i-3}{d-2}\right) 
\geq 2 C_7 2^{i(d-1)/2}
\end{align*}
using the fact that:
\begin{align*}
\binom{a+b}{c} - \binom{b}{c} &= \frac{1}{d!} ((a+b)(a+(b-1))\ldots(a+(b-c)) - b(b-1)\ldots(b-c)) \\
&\geq \frac{1}{d!} (bd(a+(b-1))(a+(b-2))\ldots(a+(b-c))) \\
&= b \binom{a+b-1}{c-1}
\end{align*}
Now, we will show that $Z^k(x,\eta) \geq Z^k(\eta,\eta)/2$. Since $d(x,\eta) < C_8 2^{-i}$, thus $x.\eta = \cos d(x,\eta) > 1 - \frac{C_8^2 2^{-2i}}{2}$. We have:
\begin{align*}
Z^k(x,\eta) &= \frac{2k+d-2}{(d-2)\omega_{d-1}} P_k^{(d-2)/2}(\cos d(x,\eta)) \\
&\geq \frac{2k+d-2}{(d-2)\omega_{d-1}} P_k^{(d-2)/2}(1) - \frac{C_8^2 2^{-2i}}{2} \sup_{x \in [0,1]} \frac{2k+d-2}{(d-2)\omega_{d-1}} \frac{d}{dx} P_k^{(d-2)/2}(x) \\
&\geq Z^k(\eta,\eta) \Big(1 - \frac{C_8^2 2^{-2i}}{2} \sup_{x \in [0,1]} \frac{\frac{d}{dx} P_k^{(d-2)/2}(x)} {P_k^{(d-2)/2}(1)}\Big)
\end{align*}
Therefore, to complete the proof, it suffices to show that we can pick $C_8$ such that for all $2^{i-1} \leq k \leq 2^{i+1}$, \begin{equation*} \frac{C_8^2 2^{-2i}}{2} \sup_{x \in [0,1]} \frac{\frac{d}{dx} P_k^{(d-2)/2}(x)} {P_k^{(d-2)/2}(1)} < \frac{1}{2} \end{equation*}But $\frac{d}{dx} P_k^a(x) = 2a P_{k-1}^{a+1}(x)$ which is another Gegenbauer polynomial, and these have a maximum at 1, and so:
\begin{align*}
\frac{C_8^2 2^{-2i}}{2} \sup_{x \in [0,1]} \frac{\frac{d}{dx} P_k^{(d-2)/2}(x)} {P_k^{(d-2)/2}(1)} &\leq \frac{C_8^2 2^{-2i}}{2} \frac{(d-2) P_{k-1}^{d/2}(1)} {P_k^{(d-2)/2}(1)} \\
&= \frac{C_8^2 2^{-2i}}{2} \frac{(d-2)\binom{k+d-2}{d-1}}{\binom{k+d-3}{d-3}} \\
&= \frac{C_8^2 2^{-2i} k(k+d-2)}{2(d-1)} \leq \frac{C_8^2 (2+\frac{d-2}{2^i})}{(d-1)} \leq 2C_8^2
\end{align*}
using the fact that $P_a^b(1) = \binom{a+2b-1}{2b-1}$ and $2^{i-1} \leq k \leq 2^{i+1}$.
\end{proof}

We end this paper with a demonstration of the lower bound for the problem of density estimation.

\subsection{Lower Bound on the Testing Problems}
A standard way of finding the lower bound for the problem of density estimation via the following testing problem:

\begin{proposition}
\label{blahblah}
Let $\psi:S^{d-1} \to \mathbb{R}$ be bounded such that $\int_S^{d-1} \psi(x)\,dx=0$ and $\int_S^{d-1} \psi(x)^2\,dx\leq1$. Further, let $X = \{X_1,\,X_2,\ldots,\,X_n\}$ be an independent and identically distributed sample of $n$ values from $f$ under the following hypotheses:

\begin{equation*}
H_0: f=\frac{1}{\omega_{d-1}},\, H_1:f=\frac{1}{\omega_{d-1}} + \frac{\psi}{2\sqrt{\omega_{d-1} n}}
\end{equation*}

Now, consider the set of all tests $\Psi:[0,1]^n \to \{0,1\}$. Then:

\begin{equation*}
\inf_{\Psi} \max_{j \in \{0,1\}} \mathbb{P}_{H_j}(\Psi \neq j) \geq \frac{1}{60}
\end{equation*}
\end{proposition}

\begin{proof} We have the following: \begin{align*} \inf_{\Psi} \max_{j \in \{0,1\}} \mathbb{P}_{H_j}(\Psi \neq j) 
&\geq \frac{1}{2} \inf_{\Psi} (\mathbb{E}_{H_0}(\Psi) + \mathbb{E}_{H_1}(1-\Psi)) \\
&\geq \frac{1}{2} \inf_{\Psi} \Bigg(\mathbb{E}_{H_0}(\Psi) + \mathbb{E}_{H_0}\bigg((1-\Psi) \cdot \frac{1}{3} \cdot 1\Big(\frac{dP_{H_1}}{dP_{H_0}} \geq \frac{1}{3}\Big)\bigg)\Bigg) \\
&\geq \frac{1}{2} \inf_{\Psi} \Bigg(\mathbb{E}_{H_0}\bigg((\Psi + (1-\Psi)) \cdot \frac{1}{3} \cdot 1\Big(\frac{dP_{H_1}}{dP_{H_0}} \geq \frac{1}{3}\Big)\bigg)\Bigg) \\
&\geq \frac{1}{6} \Bigg(1 - \mathbb{P}_{H_0}\bigg(\Big|\frac{dP_{H_1}}{dP_{H_0}} - 1\Big| > \frac{2}{3}\bigg)\Bigg) \\
&\geq \frac{1}{6} \bigg(1 - \frac{3}{2}\sqrt{\mathbb{E} \Big(\frac{dP_{H_1}}{dP_{H_0}} - 1\Big)^2}\bigg)
\end{align*} 

But using the fact that $\int_0^1 \psi(x)\,dx=0$ and $\int_0^1 \psi(x)^2\,dx\leq1$.\begin{align*} \mathbb{E} \Big(\frac{dP_{H_1}}{dP_{H_0}} - 1\Big)^2 &= \int_{S_{d-1}^n} \Bigg(\prod_{i=1}^n \bigg(1+\frac{\psi(x)\sqrt{\omega_{d-1}}}{2\sqrt{n}}\bigg) - 1\Bigg)^2\frac{1}{(\omega_{d-1})^n}\,d\mathbf{x} \\
&= \frac{1}{(\omega_{d-1})^n} \int_{S_{d-1}^n} \prod_{i=1}^n \bigg(1+\frac{\psi(x)\sqrt{\omega_{d-1}}}{2\sqrt{n}}\bigg)^2 - 2 \prod_{i=1}^n \bigg(1+\frac{\psi(x)\sqrt{\omega_{d-1}}}{2\sqrt{n}}\bigg) + 1\,d\mathbf{x} \\
&= \frac{1}{(\omega_{d-1})^n} \Bigg(\int_{S_{d-1}} \bigg(1+\frac{\psi(x)\sqrt{\omega_{d-1}}}{2\sqrt{n}}\bigg)^2\,dx\Bigg)^n - 1 \\
&\leq \Big(1 + \frac{1}{4n}\Big)^n - 1 \leq \exp\Big(\frac{1}{4}\Big) - 1 \leq 0.36
\end{align*}
Putting these two estimates together yields the result.\end{proof}

We relate the testing problem to that of density estimation in the following way:

\begin{proof}[Proof of Theorem~\ref{minimax}] For notational convenience, we only prove this for $0< t \leq 1$. We will use the definitions of global Besov spaces from \cite{decomp}. Recall that $H_k(S^{d-1})$ is the space of spherical harmonics of degree $k$ and let $\Pi_n = \bigoplus_{k=0}^n H_k(S^{d-1})$. Then we define the approximation error of a function $f$ by polynomials in this class to be:
\begin{equation*}
E_k(f) = \inf_{P \in \Pi_n} ||f - P||_\infty
\end{equation*}

and the Besov space $B^t_{\infty\infty}$:

\begin{definition} $f \in B^t_{\infty\infty}(M)$ if and only if $\sup_k k^t E_k(f) \leq M$ \end{definition}

Let $2^{i+1} \sim n^{1/(2t+d-1)}$ and consider $\psi(x) := \psi_{i\eta}(x)$. It is a polynomial of degree $2^{i+1}$, so for $k \geq 2^{i+1}$, $E_k(1 + \frac{\psi(x)}{\sqrt{n}}) = 0$. Since by Equation~\ref{psiimpt}, $||\psi||_\infty \leq c_1 2^{i(d-1)/2}$, so we can use the polynomial approximation $1$ to obtain:
\begin{equation*}
\sup_k k^t E_k\left(1 + \frac{\psi}{\sqrt{n}}\right) \leq \frac{2^{t(i+1)}}{\sqrt{n}} c_1 2^{i(d-1)/2} \leq c_1
\end{equation*} 

Thus $1 + \frac{\psi}{\sqrt{n}} \in B^t_{\infty\infty}(c_1)$ for all $i$. Now, by Theorem 5.5 in~\cite{Geller}, such functions are also in $C^t_M(S^{d-1})$ (assuming $0< t \leq 1$). Hence:

\begin{align*}
\liminf_n \inf_{\hat{f}\in\mathcal{F}_n} \sup_{f\in C_{M,\delta}^t(\eta)} \mathbb{E}_f \Big|n^{\frac{t}{2t+d-1}} (\hat{f}(X_1, X_2,\ldots,X_n) - f(\eta))\Big| \\ \geq \liminf_n \inf_{\hat{f}\in\mathcal{F}_n} \max_{f=1~\text{or}~1 + \frac{\psi}{\sqrt{n}}} \mathbb{E}_f \Big|n^{\frac{t}{2t+d-1}} (\hat{f}(X_1, X_2,\ldots,X_n) - f(\eta))\Big|
\end{align*}

Now we construct a test $\Psi_n = 1(\hat{f}(X_1, X_2,\ldots,X_n) > 1 + \frac{\psi(\eta)}{4\sqrt{n}})$. Also, $||\psi_{i\eta}||_2 \leq 1$ (Lemma~\ref{quad}) and $\int_{S^{d-1}} \psi_{i\eta} = 0$. Thus, by the Proposition~\ref{blahblah}, we have:
\begin{align*}
&\max_{j \in \{0,1\}} \bigg\{\mathbb{P}\Big(\hat{f} \leq 1 + \frac{\psi(\eta)}{4\sqrt{n}}|f = 1 + \frac{\psi}{2\sqrt{n}}\Big),\,\mathbb{P}\Big(\hat{f} > 1 + \frac{\psi(\eta)}{4\sqrt{n}}|f = 1\Big)\bigg\} \geq \frac{1}{60} \\
&\Rightarrow \max_{f=1~\text{or}~1 + \frac{\psi}{\sqrt{n}}} \mathbb{P} \bigg(|\hat{f} - f(\eta)| \geq \frac{\psi(\eta)}{4\sqrt{n}}\bigg) \geq \frac{1}{60} \\
&\Rightarrow \max_{f=1~\text{or}~1 + \frac{\psi}{\sqrt{n}}} \mathbb{E}(|\hat{f} - f(\eta)|) \geq \frac{\epsilon \psi(\eta)}{240 \sqrt{n}}
\end{align*}
Combining everything we have, and using Lemma~\ref{lastbutnotleast} above that there exists $C_7$ such that $\psi_{i\eta}(\eta) \geq C_7 2^{i(d-1)/2} \sim C_7 \frac{n^{(d-1)/2(2t+d-1)}}{2^{(d-1)/2}}$, we get:
\begin{align*}
\liminf_n \inf_{\hat{f}\in\mathcal{F}_n} \sup_{f\in C_{M,\delta}^t(\eta)} \mathbb{E}_f \Big|n^{\frac{t}{2t+d-1}} (\hat{f}(X_1, X_2,\ldots,X_n) - f(\eta))\Big| &\geq \liminf_n \frac{\epsilon \psi(\eta)}{240 \sqrt{n}} n^{\frac{t}{2t+d-1}}
\\&\geq \frac{\epsilon C_7}{240 \cdot 2^{(d-1)/2}}
\end{align*}
which is the bound needed.\end{proof}


\begin{thebibliography}{99}
\bibitem{char} \textsc{Andersson} Characterization of pointwise h\"{o}lder regularity. \textit{Applied and Computational Harmonic Analysis 4} 1997, 424-443
\bibitem{Baldi} \textsc{Baldi}, \textsc{Kerkyacharian}, \textsc{Marinucci}, \textsc{Picard} Adaptive density estimation for directional data using needlets. \textit{The Annals of Statistics Vol 39} 2009, 3362-3395
\bibitem{first2} \textsc{Donoho}, \textsc{Johnstone}, \textsc{Kerkyacharian} and \textsc{Picard} (1996) Density estimation by wavelet thresholding. \textit{The Annals of Statistics Vol 24} p.508-539
\bibitem{faraut} \textsc{Faraut}, Analysis on lie groups: an introduction. \textit{Cambridge Studies in Advanced Mathematics, Cambridge University Press} 2008, 189
\bibitem{Geller} \textsc{Geller}, \textsc{Mayeli}, Continuous wavelets on compact manifolds. \textit{Mathematische Zeitschrift 262} 2009, 895-927
\bibitem{ugh} \textsc{Gin\'e}, \textsc{Nickl}, An exponential inequality for the distribution function of the kernel density estimator, with applications to adaptive estimation. \textit{Probability Theory and Related Fields 143} 2009, 569-596.
\bibitem{nickl} \textsc{Gin\'e}, \textsc{Nickl}, Confidence bands in density estimation. \textit{The Annals of Statistics Vol 38} 2010, 1122-1170
\bibitem{cross} \textsc{Hall}, \textsc{Watson}, \textsc{Cabrera} Kernel density estimation with spherical data. \textit{Biometrika 74} 1987, 751-62
\bibitem{winner} \textsc{Hoffmann}, \textsc{Nickl}, On adaptive inference and confidence bands. \textit{The Annals of Statistics Vol 39} 2011, 2383-2409
\bibitem{Jaffard2} \textsc{Jaffard} Wavelet techniques for pointwise regularity. \textit{Annales de la Facult\'{e} des Sciences de Toulouse Vol XV} 2006, 3-33
\bibitem{nickl2} \textsc{Kerkyacharian}, \textsc{Nickl}, \textsc{Picard}, Concentration inequalities and confidence bands for needlet density estimators on compact homogeneous manifolds. \textit{Probability Theory and Related Field}, to appear.
\bibitem{plug} \textsc{Klemel\"{a}} Estimation of densities and derivatives of densities with directional data. \textit{Journal of Multivariate Analysis 73} 2000, 18-40.
\bibitem{Lepski} \textsc{Lepski}, \textsc{Mammen}, \textsc{Spokoiny}, Optimal spatial adaptation to inhomogeneous smoothness: an approach based on kernel estimates with variable bandwidth selectors. \textit{The Annals of Statistics Vol 25} 1997, 929-947. 
\bibitem{Low} \textsc{Low}, On nonparametric confidence intervals. \textit{The Annals of Statistics Vol 25} 1997, 2547-2554
\bibitem{decomp} \textsc{Narcowich}, \textsc{Petrushev}, \textsc{Ward} Decomposition of Besov and Triebel-Lizorkin spaces on the sphere. \textit{Journal of Functional Analysis 238} 2006, 530-564
\bibitem{local} \textsc{Narcowich}, \textsc{Petrushev}, \textsc{Ward}, Localized tight frames on spheres. \textit{Siam J Math. Anal. Vol 38} 2006, 574-594
\bibitem{more} \textsc{Picard}, \textsc{Tribouley}, Adaptive confidence interval for pointwise curve estimation. \textit{The Annals of Statistics Vol 28} 2000, 298-335
\bibitem{Stein} \textsc{Stein}, \textsc{Weiss} Introduction to Fourier Analysis on Euclidean Spaces \textit{Princeton} 1971, 137-150.

\end{thebibliography}
\end{document}